\def\year{2022}\relax
\documentclass[letterpaper]{article} 
\usepackage{aaai22}  
\usepackage{times}  
\usepackage{helvet}  
\usepackage{courier}  
\usepackage[hyphens]{url}  
\usepackage{graphicx} 
\usepackage{natbib}  
\usepackage{caption} 
\DeclareCaptionStyle{ruled}{labelfont=normalfont,labelsep=colon,strut=off} 
\usepackage{algorithm}
\usepackage[noend]{algorithmic}
\usepackage{booktabs}
\usepackage{amsthm}
\usepackage{thmtools, thm-restate}
\usepackage{multirow}

\usepackage{newfloat}
\usepackage{listings}
\floatstyle{ruled}
\newfloat{listing}{tb}{lst}{}
\floatname{listing}{Listing}
\title{Efficient Decentralized Learning Dynamics for Extensive-Form Coarse Correlated Equilibrium: No Expensive Computation of Stationary Distributions Required}
\author{
	Gabriele Farina\thanks{Equal contribution.},\textsuperscript{\rm 1} Andrea Celli$^\ast$,\textsuperscript{\rm 2} Tuomas Sandholm\textsuperscript{\rm 1,3,4,5}
}
\affiliations{
     \textsuperscript{\rm 1}Carnegie Mellon University, \textsuperscript{\rm 2}Bocconi University, \textsuperscript{\rm 3}Strategic Machine, Inc., \textsuperscript{\rm 4}Strategy Robot, Inc., \textsuperscript{\rm 5}Optimized Markets, Inc.\\
     gfarina@cs.cmu.edu, andrea.celli2@unibocconi.it, sandholm@cs.cmu.edu
}

\usepackage{bm}
\usepackage{amsmath}
\usepackage{xparse}
\usepackage{mleftright}
\usepackage{dsfont}
\usepackage[capitalise,noabbrev]{cleveref}
\usepackage[mathscr]{eucal}
\usepackage{graphicx}
\usepackage{nicefrac}
\usepackage{amssymb}
\usepackage{mathtools}

\usepackage[normalem]{ulem}

\crefname{problem}{Problem}{Problems}

\crefname{assumption}{Assumption}{Assumptions}

\newtheorem{lemma}{Lemma}

\newtheorem{proposition}{Proposition}
\newtheorem{definition}{Definition}

\usepackage{tikz}
\usetikzlibrary{calc,arrows.meta,arrows}
\tikzset{cross/.style={path picture={
  \draw[black]
(path picture bounding box.south east) -- (path picture bounding box.north west) (path picture bounding box.south west) -- (path picture bounding box.north east);
}}}
\tikzstyle{chanode}=[fill=white,draw=black,circle,cross,inner sep=.8mm]
\tikzstyle{pl1node}=[fill=black,draw=black,circle,inner sep=.5mm]
\tikzstyle{pl2node}=[fill=white,draw=black,circle,inner sep=.55mm]
\tikzstyle{termina}=[fill=white,draw=black,inner sep=.6mm]
\tikzstyle{decpt}  =[fill=black,draw=black,inner sep=.8mm]
\tikzstyle{obspt}  =[fill=white,draw=black,cross,inner sep=0.95mm]

\usetikzlibrary{matrix,backgrounds,positioning}
\pgfdeclarelayer{back}
\pgfsetlayers{back,background,main}
\NewDocumentCommand{\fhighlight}{O{black!15}mm}{%
\fill[#1] (#2.north west) rectangle (#3.south east);
}

\renewcommand{\vec}[1]{\bm{#1}}
\newcommand{\defeq}{\coloneqq}
\newcommand{\emptyseq}{\varnothing}
\newcommand{\Hist}{\mathscr{H}}
\newcommand{\Z}{\mathscr{Z}}
\newcommand{\A}{\mathscr{A}}
\newcommand{\cR}{\mathcal{R}}
\newcommand{\cX}{\mathcal{X}}
\newcommand{\cY}{\mathcal{Y}}

\newcommand{\muv}{\vec{\mu}}

\DeclareMathOperator{\co}{co}

\newcommand{\bbR}{\mathbb{R}}

\newcommand{\bbN}{\mathbb{N}_{>0}}

\NewDocumentCommand{\pure}{O{i}}{%
\ifthenelse{\equal{#1}{}}%
    {\vec{\pi}}%
    {\vec{\pi}^{(#1)}}%
}
\NewDocumentCommand{\hatpure}{O{i}}{%
\ifthenelse{\equal{#1}{}}%
    {\hat{\vec{\pi}}}%
    {\hat{\vec{\pi}}^{(#1)}}%
}
\NewDocumentCommand{\puret}{O{i}O{t}}{\vec{\pi}^{(#1),\,#2}}
\NewDocumentCommand{\Pure}{O{i}}{\Pi^{(#1)}}
\NewDocumentCommand{\tdev}{O{\hatinfo}O{\hatpure[]}}{%
    \phi_{#1\to#2}
}
\NewDocumentCommand{\Mdev}{O{\hatinfo}O{\hatpure[]}}{%
    \vec{M}_{#1\to#2}
}
\NewDocumentCommand{\cJ}{O{i}}{\mathscr{J}^{(#1)}}
\NewDocumentCommand{\cI}{O{i}}{\mathscr{I}^{(#1)}}
\newcommand{\info}{I}

\newcommand{\lambdav}{\vec{\lambda}}
\newcommand{\hatinfo}{\hat{\info}}
\NewDocumentCommand{\Info}{O{i}}{\mathscr{I}^{(#1)}}

\NewDocumentCommand{\Seqs}{O{i}}{\Sigma^{(#1)}}
\NewDocumentCommand{\pc}{O{z}}{p^{(c)}(#1)}
\NewDocumentCommand{\ut}{O{i}}{u^{(#1)}}
\NewDocumentCommand{\bbone}{O{XXX}}{\mathds{1}_{\{#1\}}}
\NewDocumentCommand{\Ph}{O{i}}{\Psi^{(#1)}}
\NewDocumentCommand{\Tph}{O{i}}{(\co\Ph)}
\NewDocumentCommand{\Phj}{O{i}O{\hat\sigma}}{\bar{\Psi}^{(#1)}_{#2}}

\NewDocumentCommand{\Q}{O{i}}{\mathcal{Q}^{(#1)}}
\NewDocumentCommand{\q}{}{\vec{q}}
\NewDocumentCommand{\qt}{O{i}O{t}}{\vec{q}^{(#1),\,#2}}
\NewDocumentCommand{\seq}{m}{\texttt{#1}}

\NewDocumentCommand{\parseq}{O{i}O{\info}}{\sigma^{(#1)}(#2)}

\NewDocumentCommand{\rdev}{O{i}O{\hatinfo}O{\hat{\vec{\pi}}}}{%
    r^{(#1)}_{\muv,#2\to#3}
}
\NewDocumentCommand{\cumr}{O{i}O{T}}{R^{(#1),\,#2}}

\begin{document}

\maketitle

\begin{abstract}

While in two-player zero-sum games the Nash equilibrium is a well-established prescriptive notion of optimal play, its applicability as a prescriptive tool beyond that setting is limited. Consequently, the study of decentralized learning dynamics that guarantee convergence to \emph{correlated} solution concepts in multiplayer, general-sum extensive-form (\emph{i.e.}, tree-form) games has become an important topic of active research. 
The per-iteration complexity of the currently known learning dynamics depends on the specific correlated solution concept considered. For example, in the case of extensive-form correlated equilibrium (EFCE), all known dynamics require, as an intermediate step at each iteration, to compute the stationary distribution of multiple Markov chains, an expensive operation in practice. Oppositely, in the case of normal-form coarse correlated equilibrium (NFCCE), simple no-external-regret learning dynamics that amount to a linear-time traversal of the tree-form decision space of each agent suffice to guarantee convergence.  
This paper focuses on extensive-form coarse correlated equilibrium (EFCCE), an intermediate solution concept that is a subset of NFCCE and a superset of EFCE. Being a superset of EFCE, any learning dynamics for EFCE automatically guarantees convergence to EFCCE. However, since EFCCE is a simpler solution concept, this begs the question: \emph{do learning dynamics for EFCCE that avoid the expensive computation of stationary distributions exist?} This paper answers the previous question in the positive. Our learning dynamics only require the orchestration of no-external-regret minimizers, thus showing that EFCCE is more akin to NFCCE than to EFCE from a learning perspective. Our dynamics guarantees that the empirical frequency of play after $T$ iteration is a $O(1/\sqrt{T})$-approximate EFCCE with high probability, and an EFCCE almost surely in the limit.
%


\end{abstract}

\section{Introduction}

In a normal-form game (i.e., a game with simultaneous moves), a \emph{correlated strategy} is defined as a probability distribution over joint action profiles, and it is customarily modeled via a trusted external mediator that draws an action profile from this distribution, and privately recommends to each player their component. A correlated strategy is a \emph{correlated equilibrium} (CE) if, for each player, the mediator's recommendation is the best action in expectation, assuming all the other players follow their recommended actions~\cite{aumann1974subjectivity}.
CE is an appealing solution concept in real-world strategic interactions involving more than two players with arbitrary (i.e., general-sum) utilities. Indeed, in those settings, the notion of CE overcomes several weaknesses of the \emph{Nash equilibrium} (NE)~\cite{nash1950equilibrium}. In particular, in settings beyond two-players zero-sum games, the NE is prone to equilibrium selection issues, it is computationally intractable (being PPAD-complete even in two-player games~\cite{chen2006settling,daskalakis2009complexity}), and the social welfare that can be attained at an NE may be arbitrarily lower than what can be achieved through a CE~\cite{koutsoupias1999worst,roughgarden2002bad,celli2018}. In contrast, a CE explicitly models synchronization between players, and it is computable in polynomial time in normal-form games.
Moreover, in arbitrary normal-form games, the notion of CE arises naturally from simple decentralized learning dynamics~\cite{foster1997calibrated,hart2000simple}. Decentralized learning dynamics offer a parallel, scalable avenue for computing equilibria, and allow players to circumvent the---often unreasonable---assumption that they have perfect knowledge of other players' payoff functions. In particular, players can adjust their strategies on the basis of their own private payoff function, and on the observed behavior of the other players. In the case of NE, decentralized learning dynamics are only known in the two-player zero-sum setting (see, e.g.,~\citet{Cesa-Bianchi06:Prediction,Hart03:Uncoupled}).

Extensive-form games generalize normal-form games by modeling both sequential and simultaneous moves, as well as imperfect information. Because of their sequential nature, extensive-form games admit various notions of correlated equilibrium, which essentially differ in the time at which each player can decide whether to deviate or to follow recommendations. Three natural extensions of CE to extensive-form games are the \emph{extensive-form correlated equilibrium} (EFCE) by~\citet{von2008extensive}, the \emph{extensive-form coarse correlated equilibrium} (EFCCE) by~\citet{farina2019coarse}, and the \emph{normal-form coarse correlated equilibrium} (NFCCE) by~\citet{celli2018computing}. The set of those equilibria are such that, for any extensive-form game, EFCE $\subseteq$ EFCCE $\subseteq$ NFCCE.
Decentralized no-regret learning dynamics are known for the set of EFCE~\cite{Celli20:NoRegret,farina2021simple,morrill2021efficient}, and they require, as an intermediate step at each iteration, to compute the stationary distribution of multiple Markov chains, which can be an expensive operation in practice.  On the other hand, the set of NFCCE admits simple no-external-regret learning dynamics that amount to a linear-time traversal of the tree-form decision space of each agent~\cite{celli2019learning}. 
This paper studies decentralized learning dynamics converging to the set of EFCCE. In an EFCCE, before the beginning of the game, the mediator draws a recommended action for each of the possible information sets that players may encounter in the game, according to some known probability distribution defined over joint deterministic strategies. These recommendations are not immediately revealed to each player. Instead, the mediator incrementally reveals relevant action recommendations as players reach new information sets.
At each information set the acting player has to commit to following the recommended move before it is revealed to them, by only knowing the mediator's policy used to draw recommendations and the past recommendations issued from the root of the game tree down to the current information set~\cite{farina2019coarse}. If the acting player decides to deviate (i.e., commits to \emph{not following} the recommendation), their recommendations will no longer be issued by the mediator.
Since the set of EFCEs is a subset of the set of EFCCEs~\citep{farina2019coarse}, learning dynamics for EFCE automatically guarantees convergence to EFCCE. However, since EFCCE is a simpler solution concept, the following natural question arises: \emph{do learning dynamics for EFCCE that avoid the expensive computation of stationary distributions exist?} This paper answers the previous question in the positive. 
In particular, we define the notion of \emph{coarse trigger regret} as a particular instantiation of the \emph{phi-regret minimization} framework~\cite{Greenwald03:General,Stoltz07:Learning,Gordon08:No}, and we show that if each player behaves according to a no-coarse-trigger-regret algorithm, then the empirical frequency of play approaches the set of EFCCEs. Then, we provide an efficient algorithm for minimizing coarse trigger regret based on the general template for constructing phi-regret minimizers by~\citet{Gordon08:No}. We show that, in contrast to EFCE, any convex combination of coarse trigger deviation functions admits a fixed point strategy which can be computed in closed form, without requiring to compute the stationary distribution of any Markov chain. 
In particular, our learning dynamics only require the orchestration of no-external-regret minimizers, thus showing that EFCCE is more akin to NFCCE than to EFCE from a learning perspective. Our algorithm guarantees that the empirical frequency of play after $T$ iteration is a $O(1/\sqrt{T})$-approximate EFCCE with high probability, and an EFCCE almost surely in the limit.

\noindent\textbf{Related work.}\
The study of adaptive procedures converging to a CE in normal-form games dates back to the works by~\citet{foster1997calibrated},~\citet{fudenberg1995consistency,fudenberg1999conditional}, and~\citet{hart2000simple,hart2001general}.
In more recent years, a growing effort has been devoted to understanding the relationships between no-regret learning dynamics and equilibria in extensive-form games.
While in two-player zero-sum extensive-form games it is widely known that no-regret learning dynamics converge to an NE (see, e.g.,~\cite{zinkevich2008regret,tammelin2015solving,lanctot2009monte,brown2019solving}) the general case of multi-player general-sum games is less understood.
\citet{celli2019learning} provide variations of the classical CFR algorithm, showing that they provably converge to the set of NFCCEs.
\citet{Celli20:NoRegret} describe learning dynamics that converge to the set of EFCE almost surely in the limit. Their algorithm requires to instantiate and manage a number of \emph{internal regret minimizers} growing linearly in the number of information sets in the game. Each internal regret minimizer internally requires the computation of a stationary distribution of a Markov chain~\citep{Cesa-Bianchi06:Prediction,blum2007external}. \citet{farina2021simple} extend the work by~\citet{Celli20:NoRegret}, giving convergence guarantees to the set of EFCEs at finite time in high probability. The latter paper operates within the phi-regret minimization framework of~\citet{Gordon08:No}, and requires the computation of the stationary distribution of multiple Markov chains at each iteration.
The recent work by~\citet{morrill2021efficient} presents a general framework for achieving hindsight rational learning~\cite{morrill2020hindsight} in extensive-form games for various types of behavioral deviations. It is known that, when framework by~\citet{morrill2021efficient} (EFR) is instantiated with different choices of sets of behavioral deviations, EFR leads to different solution concepts (including EFCCE in the case of \emph{blind causal deviations}). Just like the other mentioned approaches, the EFR framework requires the computation of fixed points of linear transformations at each iteration.
We conjecture that a similar result as this paper (\emph{i.e.}, the existence of a fixed point that can computed in closed form without the need to compute any stationary distribution of a Markov chain) could also be derived within the EFR framework, when blind causal deviations are considered, though we leave exploration of that direction open.


%
%
%

\section{Preliminaries}\label{sec:preliminaries}


\noindent 
The set $\{1,\ldots,n\}$, with $n\in\bbN$, is compactly denoted as $[n]$. Given a set $S$, we denote its convex hull with the symbol $\co S$. 

\subsection{Extensive-Form Games}

An extensive-form game is usually defined by means of an oriented rooted game tree. The set of nodes that are not a leaf of the game tree is denoted by $\Hist$. Each node $h\in\Hist$ is called a \emph{decision node} and has associated a player that acts at that node by choosing one action from the set of available actions at $h$, which we denote by $\A(h)$.
In an $n$-player extensive-form game, the set of players is the set $[n]\cup\{c\}$, where $c$ denotes the chance player, which is a fictitious player that selects actions according to fixed probability distributions representing exogenous stochasticity of the environment (e.g., a roll of the dice).
%
Leaves of the game tree are called \emph{terminal nodes}, and represent the outcomes of the game; their set of available actions is conventionally set to $\emptyset$ and they are not assigned to an acting player. The set of such nodes is denoted by $\Z$. When the game transitions to a terminal node $z\in\Z$, payoffs are assigned to each non-chance player according to the set of payoff functions $\{u^{(i)}:\Z\to\bbR\}_{i\in[n]}$. 
Moreover, we let $p^{(c)}:\Z\to(0,1)$ denote the function assigning to each terminal node $z$ the product of probabilities of chance moves encountered on the path from the root of the game tree to $z$.

\noindent \textbf{Imperfect information.} 
The set of decision nodes of each player $i\in[n]$ is partitioned into a collection $\Info$ of sets of nodes, called \emph{information sets}. Each information set $\info\in\Info$ groups together nodes that Player $i$ cannot distinguish between when Player $i$ acts. Therefore, we have that $\A(h)=\A(h')$ for any pair of nodes $h,h'\in \info$. Then, we can safely write $\A(\info)$ to indicate the set of actions available at any decision node belonging to $\info\in\Info$.
As it is customary in the literature, we assume that the extensive-form game has \emph{perfect recall}, that is, information sets are such that no player forgets information once acquired. This means that, for any player $i\in[n]$ and any two nodes $h,h'\in \info$, with $\info\in\Info$, the sequence of Player $i$'s actions from the root to $h$ must coincide with the sequence of Player $i$'s actions from the root to $h'$. Therefore, for any $i\in[n]$, we can define a partial ordering $\prec$ on $\Info$ as follows: for any $\info,\info'\in\Info$, $\info\prec \info'$
if there exist nodes $h'\in \info'$ and $h\in \info$ such that the path from the root of the game to $h'$ passes through $h$. An immediate consequence of perfect recall is that for any $i\in[n]$, $\Info$ is well-ordered by $\prec$ (i.e., given $\info\in\Info$, the set of its predecessors forms a chain).

\noindent \textbf{Sequences.} For any player $i\in[n]$, information set $\info\in\Info$, and action $a\in\A(\info)$, we denote by $\sigma=(\info,a)$ the sequence of Player $i$'s actions on the path from the root of the game tree down to action $a$ (included) taken at any decision node in information set $\info$. We denote by $\emptyseq$ the \emph{empty sequence} of Player $i$. Then, the set of Player $i$'s sequences is defined as $\Seqs\defeq\{(\info,a):\info\in\Info, a\in\A(\info)\}\cup\{\emptyseq\}$. Given an information set $\info\in\Info$, we denote by $\parseq\in\Seqs$ the \emph{parent sequence} of $\info$, that is, the last sequence encountered by Player $i$ on the path from the root of the game tree to any node in $\info$. Whenever $\parseq=(\info',a)$, we say that $\info$ is \emph{immediately reachable} from sequence $\parseq$. If Player $i$ never acts before $\info$, then $\parseq=\emptyseq$, and we say that information set $\info$ is a \emph{root information set} of Player $i$. Moreover, for any $z\in\Z$, $\parseq[i][z]\in\Seqs$ is the last sequence of Player $i$'s actions encountered on the path from the root of the game tree to terminal node $z$. We let $\parseq[i][z]=\emptyseq$ if Player $i$ never plays on the path from the root to $z$. 
Analogously to what we did for information sets, we introduce a partial ordering on sequences: for every $i\in[n]$, and any pair $\sigma,\sigma'\in\Seqs$, the relation $\sigma\prec\sigma'$ holds if $\sigma=\emptyseq\neq\sigma'$, or if the sequences are such that $\sigma=(\info,a)$, $\sigma'=(\info',a')$, and the set of Player $i$'s actions on the path from the root to $\info'$ include playing action $a$ at one node belonging to $\info$.
For any $i\in[n]$, $\sigma\in\Seqs$, and $\info\in\Info$, we write $\sigma\succ \info$ to mean that the sequence of Player $i$'s actions $\sigma$ must lead the player to pass through $\info$, formally $\sigma=(\info',a')\in\Seqs\setminus\{\emptyseq\}\wedge \info'\succ \info$. Moreover, for $\sigma\in\Seqs$ and $\info\in\Info$, we write $\sigma\succeq I$ when $\sigma\succ I$ or $\sigma=(I,a)$.
Then, we let $\Seqs_\info\defeq\{\sigma\in\Seqs:\sigma\succeq \info\}\subseteq\Seqs$ be the set of Player $i$'s sequences that terminate at $\info$ or any of its descendant information sets, and $\Z^{(i)}_\info\defeq \mleft\{z\in\Z:\parseq[i][z]\succeq \info\mright\}$ be the set of terminal nodes reachable from information set $\info\in\Info$.

\noindent\textbf{Sequence-form strategies.} A \emph{sequence-form strategy} for Player $i\in[n]$ is a vector $\q\in\mathbb{R}^{|\Seqs|}_{\ge0}$ such that each entry $\q[(\info,a)]$ specifies the product of the probabilities of playing all of Player $i$'s actions on the path from the root down to action $a$ at information set $\info$ (included) \cite{Koller96:Efficient,Romanovskii62:Reduction,Stengel96:Efficient}. The set of valid sequence-form strategies for Player $i$ is defined by some linear probability-mass-conservation constraints. Formally, 
\begin{definition}\label{def:seq_polytope}
The \emph{sequence-form strategy polytope} for Player $i\in[n]$ is the convex polytope
$\Q\defeq\{\q\in\mathbb{R}^{|\Seqs|}_{\ge0}:\q[\emptyseq]=1,\text{\normalfont and } \q[\parseq]=\sum_{a\in\A(\info)}\q[(\info,a)],\forall \info\in\Info\}.$
\end{definition}
\noindent We let 
$\Q_\info\defeq\{\q\in\mathbb{R}^{|\Seqs_\info|}_{\ge0}:\sum_{a\in\A(j)}\q[(\info,a)]=1,\text{\normalfont and }\q[\parseq[i][\info']]=\sum_{a\in\A(\info')}\q[(\info',a)],\forall \info'\succ \info\}$
be the set of sequence form strategies only specifying Player $i$'s behavior at information set $\info$ and all of its descendant. 
The set of \emph{deterministic sequence-form strategies} for Player $i\in[n]$ is defined as $\Pure\defeq\Q\cap\{0,1\}^{|\Seqs|}$, and the set of \emph{deterministic sequence-form strategies for the subtree rooted at $\info$} is $\Pure_\info\defeq\Q_\info\cap\{0,1\}^{|\Seqs_\info|}$. 
Kuhn's Theorem implies that, for any $i\in[n]$, $\Q=\co\Pure$, and $\Q_\info=\co\Pure_\info$ for any $\info\in\Info$~\cite{kuhn1953}.
%
We denote as $\Pi\defeq\bigtimes_{i\in[n]}\Pure$ the set of \emph{joint} deterministic sequence-form strategies of all the players. Moreover, $\pure[-i]\in\bigtimes_{j\ne i}\Pure[j]$ is a tuple specifying one deterministic sequence form strategy for each player other than $i$. 
It is often useful to express Player $i$'s payoff function as a function of joint deterministic sequence-form strategy profiles belonging to $\Pi$. With a slight abuse of notation let $\ut:\Pi\to\bbR$ be such that, for each $\pure[]=(\pure[1],\ldots,\pure[n])\in\Pi$,
\[
\ut(\pure[])\defeq \sum_{z\in\Z}p^{(c)}(z) \ut(z)\bbone[\pure[i][\parseq[i][z]]=1\forall i\in[n]].
\]

\subsection{Regret Minimization and Phi-Regret Minimization}\label{sec:preliminaries regret minimization}

A \emph{regret minimizer} for a set $\cX$ is an abstract model for a decision maker that repeatedly interacts with a black-box environment. At each time $t$, a regret minimizer provides two operations: (i) \textsc{NextElement} will make the regret minimizer output an element $\vec{x}^t\in\cX$; (ii) \textsc{ObserveUtility}$(\ell^t)$ will inform the regret minimizer of the environment's feedback in the form of a linear utility function $\ell^t:\cX\to\bbR$ which may depend adversarially on past choices $\vec{x}^1,\ldots,\vec{x}^{t-1}$ of the regret minimizer.
At each $t$, the regret minimizer will output a decision $\vec{x}^t$ on the basis of previous outputs $\vec{x}^1,\ldots,\vec{x}^{t-1}$ and corresponding observed utility functions $\ell^1,\ldots,\ell^{t-1}$. However, no information about future losses is available to the decision maker. 
The performance of a regret minimizer is usually evaluated in terms of its \emph{cumulative regret}
\begin{equation}\label{eq:external regret}
R^T\defeq\max_{\vec{x}^\ast\in\cX}\sum_{t=1}^T\mleft(\ell^t(\vec{x}^\ast)-\ell^t(\vec{x}^t)\mright).
\end{equation}
The cumulative regret represents how much Player $i$ would have gained by always playing the best action in hindsight, given the history of utility functions observed up to iteration $T$.
Then, the objective is to guarantee a cumulative regret growing asymptotically sublinearly in the time $T$. For example, various regret minimizers guarantee a cumulative regret $R^T=O(\sqrt{T})$ at all times $T$ for any convex and compact set $\cX$ (see, e.g.,~\citet{Cesa-Bianchi06:Prediction}).

A \emph{phi-regret minimizer}~\cite{Stoltz07:Learning,Greenwald03:General} is a generalization of the notion of regret minimizer which can be defined as follows.
\begin{definition}\label{def:phi regret minimizer}
Given a set $\cX$ of points and a set $\Phi$ of linear transformations $\phi:\cX\to\cX$, a \emph{phi-regret minimizer relative to $\Phi$ for the set $\cX$}---abbreviated \emph{``$\Phi$-regret minimizer''}---is an object with the same semantics and operations of a regret minimizer, but whose quality metric is its \emph{cumulative phi-regret relative to $\Phi$} (or \emph{$\Phi$-regret} for short)
\begin{equation}\label{eq:cum phi regret}
    R^T \defeq \max_{\phi^* \in \Phi} \sum_{t=1}^T \Big( \ell^t(\phi^*(\vec{x}^t)) - \ell^t(\vec{x}^t) \Big).
\end{equation}
The goal for a phi-regret minimizer is to guarantee that its phi-regret grows asymptotically sublinearly in $T$.
\end{definition}

We observe that a regret minimizer is a special case of a phi-regret minimizer as the cumulative regret defined in Equation~\eqref{eq:external regret} can be obtained from Equation~\eqref{eq:cum phi regret} by setting $\Phi =\{\cX \ni \vec{x} \mapsto \hat{\vec{x}}: \hat{\vec{x}}\in\cX\}$. 

A general construction by~\citet{Gordon08:No} gives a way to construct a $\Phi$-regret minimizer for $\cX$ starting from any standard regret minimizer for the set of functions $\Phi$. Specifically, let $\cR_\Phi$ be a deterministic regret minimizer for the set of transformations $\Phi$ whose cumulative regret grows sublinearly, and assume that every $\phi\in\Phi$ admits a fixed point $\phi(\vec{x}) = \vec{x}\in\cX$. Then, a $\Phi$-regret minimizer $\cR$ can be constructed starting from $\cR_\Phi$ as follows:
\begin{itemize}
    \item Each call to $\cR.\textsc{NextElement}$ first calls \textsc{NextElement} on $\cR_\Phi$ to obtain the next transformation $\phi^t$. Then, a fixed point $\vec{x}^t = \phi^t(\vec{x}^t)$ is computed and output.
    \item Each call to $\cR.\textsc{ObserveUtility}(\ell^t)$ with linear utility function $\ell^t$ constructs the linear utility function $L^t: \phi \mapsto \ell^t(\phi(\vec{x}^t))$, where $\vec{x}^t$ is the last-output strategy, and passes it to $\cR_\Phi$ by calling $\cR_\Phi.\textsc{ObserveUtility}(L^t)$.
\end{itemize}

\section{Coarse Trigger Regret and Relationship with EFCCE}
\label{sec:efcce}

In this section we describe the notion of \emph{coarse trigger deviation function} building on an idea by \citet[Section 3]{Gordon08:No}. Then, we use this notion to formally characterize the set of EFCCEs, and to define the notion of \emph{coarse trigger regret minimizer} as an instance of a phi-regret minimizer. Finally, we establish a formal connection between the set of EFCCEs and the behavior of agents minimizing their coarse trigger regret. 

\subsection{Coarse Trigger Deviation Functions}\label{sec: deviation functions}

For any $i\in[n]$, information set $\hatinfo\in\Info$, and $\hatpure[]\in\Pure_{\hatinfo}$, a \emph{coarse trigger deviation function for $\hatinfo$ and $\hatpure[]$} is a linear function which manipulates $|\Seqs|$-dimensional vectors so that any deterministic sequence form strategy that do not lead Player $i$ down to $\hatinfo$ is left unmodified. On the other hand, if a deterministic sequence form strategy prescribes Player $i$ to pass through $\hatinfo$, then its behavior at $\hatinfo$ and all of its descendant information sets is replaced with the behavior specified by the continuation strategy $\hatpure[]$.\footnote{Our definition of coarse trigger deviation function can be seen as the sequence-form counterpart to the \emph{blind causal behavioral deviations} defined by~\citet{morrill2021efficient}.
}

\begin{definition}[Coarse Trigger Deviation Function]\label{def: coarse trigger dev func}
Given an information set $\hatinfo\in\Info$, and a continuation strategy $\hatpure[]\in\Pure_{\hatinfo}$, we say that a linear function $f:\bbR^{|\Seqs|}\to\bbR^{|\Seqs|}$ is a \emph{coarse trigger deviation function corresponding to information set $\hatinfo$ and continuation strategy $\hatpure[]$} if the following two conditions hold:
\begin{itemize}
    \item $f\mleft(\pure[]\mright)=\pure[]$, for all $\pure[]\in\Pure:\pure[]\mleft[\parseq[i][\hatinfo]\mright]=0$;
    
    \item  for any $\sigma\in\Seqs$, and $\pure[]\in\Pure:\pure[]\mleft[\parseq[i][\hatinfo]\mright]=1$, 
    \begin{equation*}
        f(\pure[])[\sigma] = \begin{cases}
            \pure[]{}[\sigma] & \text{if }\sigma \not\succeq \hatinfo\\
            \hatpure[]{}[\sigma] & \text{if }\sigma\succeq \hatinfo\\
        \end{cases}.
    \end{equation*}
\end{itemize}
\end{definition}

For any $\hatinfo\in\Info$ and $\hatpure[]\in\Pure_{\hatinfo}$, it is useful to instantiate a coarse trigger deviation function in the form of a linear map $\tdev:\bbR^{|\Seqs|}\ni\pure[]\mapsto \Mdev\pure[]$, where $\Mdev\in\bbR^{|\Seqs|\times|\Seqs|}_{\ge 0}$ is the matrix such that, for any $\sigma_r,\sigma_c\in\Seqs$,
\begin{equation*}\label{eq:Mdev}
    \Mdev{}[\sigma_r, \sigma_c] = \begin{cases}
        1 & \text{if } \sigma_c \not\succeq \hatinfo \text{ and } \sigma_r = \sigma_c \\
        \hatpure[][\sigma_r] & \text{if } \sigma_c = \parseq[i][\hatinfo] \text{ and } \sigma_r \succeq \hatinfo\\
        0 & \text{otherwise}
    \end{cases}.
\end{equation*}
As a simple example of how such linear mappings are built is given in in Figure~\ref{fig:example_mapping}, where it is reported the matrix corresponding to $\tdev[\textsc{b}][\hatpure[]]$ with $\hatpure[]$ being the continuation strategy corresponding to always playing $\seq{4}$ at information set \textsc{b}.

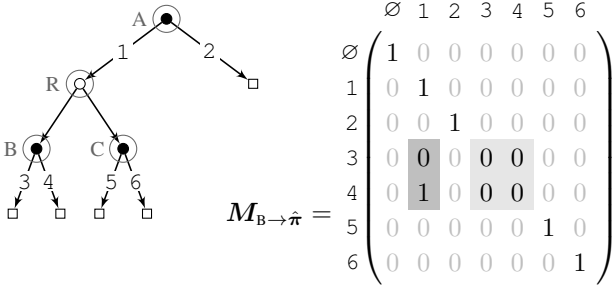
\begin{figure}[th]
\raisebox{4mm}{\begin{minipage}{0.22\textwidth}
        \begin{tikzpicture}[>=latex',baseline=0pt,scale=.9]
    \def\done{.8*1.6}
    \def\dtwo{.40*1.6}
    \def\dleaf{.22*1.6}
    \def\dvert{-.8*1.2}

    \node[fill=black,draw=black,circle,inner sep=.5mm] (A) at (0, 0) {};
    \node[fill=white,draw=black,inner sep=.6mm] (l0) at ($(\done, \dvert)$) {};
    \node[fill=white,draw=black,circle,inner sep=.5mm] (X) at ($(-\done,\dvert)$) {};
    \node[fill=black,draw=black,circle,inner sep=.5mm] (B) at ($(X) + (-\dtwo, \dvert)$) {};
    \node[fill=black,draw=black,circle,inner sep=.5mm] (C) at ($(X) + (\dtwo, \dvert)$) {};
    \node[fill=white,draw=black,inner sep=.6mm] (l1) at ($(B) + (-\dleaf, \dvert)$) {};
    \node[fill=white,draw=black,inner sep=.6mm] (l2) at ($(B) + (\dleaf, \dvert)$) {};
    \node[fill=white,draw=black,inner sep=.6mm] (l3) at ($(C) + (-\dleaf, \dvert)$) {};
    \node[fill=white,draw=black,inner sep=.6mm] (l4) at ($(C) + (\dleaf, \dvert)$) {};
    
\iftrue
    \draw[semithick] (A) edge[->] node[fill=white,inner sep=.9] {\small\seq{1}} (X);
    \draw[semithick] (A) edge[->] node[fill=white,inner sep=.9] {\small\seq{2}} (l0);
    \draw[->,semithick] (B) --node[fill=white,inner sep=.9] {\small\seq{3}} (l1);
    \draw[->,semithick] (B) --node[fill=white,inner sep=.9] {\small\seq{4}} (l2);
    \draw[->,semithick] (C) --node[fill=white,inner sep=.9] {\small\seq{5}} (l3);
    \draw[->,semithick] (C) --node[fill=white,inner sep=.9] {\small\seq{6}} (l4);

    \draw[->,semithick] (X) -- (B);
    \draw[->,semithick] (X) -- (C);
    
    \draw[black!60!white] (X) circle (.2);
    \node[black!60!white]  at ($(X) + (-.4, 0)$) {\textsc{r}};

    \draw[black!60!white] (A) circle (.2);
    \node[black!60!white]  at ($(A) + (-.4, 0)$) {\textsc{a}};
    
    \draw[black!60!white] (B) circle (.2);
    \node[black!60!white]  at ($(B) + (-.38, 0)$) {\textsc{b}};
    
    \draw[black!60!white] (C) circle (.2);
    \node[black!60!white]  at ($(C) + (-.38, 0)$) {\textsc{c}};
\else
    \draw[semithick] (A) edge[->]  (X);
    \draw[->,semithick] (A) -- (Y);
    \draw[->,semithick] (B) -- (l1);
    \draw[->,semithick] (B) -- (l2);
    \draw[->,semithick] (C) -- (l3);
    \draw[->,semithick] (C) -- (l4);
    \draw[->,semithick] (D1) -- (l5);
    \draw[->,semithick] (D1) -- (l7);
    \draw[->,semithick] (D2) --(l8);
    \draw[->,semithick] (D2) -- (l10);
    \draw[->,semithick] (X) -- (B);
    \draw[->,semithick] (X) -- (C);
    \draw[->,semithick] (Y) -- (D1);
    \draw[->,semithick] (Y) -- (D2);
\fi

\end{tikzpicture}
\end{minipage}}
\begin{minipage}{0.22\textwidth}
        \begin{tabular}{p{4cm}}
               \hspace*{-1.2cm} $\Mdev[\textsc{b}][\hatpure[]]=$\begin{tikzpicture}[baseline=-\the\dimexpr\fontdimen20\textfont2\relax ]
            \tikzset{every left delimiter/.style={xshift=1.5ex},
                     every right delimiter/.style={xshift=-1ex}};
            \matrix [matrix of math nodes,left delimiter=(,right delimiter=),row sep=.003cm,column sep=.005cm,color=black!25](m)
            {
            |[black]| 1 &           0 & 0 & 0 & 0 & 0 & 0 \\
                      0 & |[black]| 1 & 0 & 0 & 0 & 0 & 0 \\
                      0 & 0 & |[black]| 1 & 0 & 0 & 0 & 0 \\
                      0 & |[black]| 0 & 0 & |[black]| 0 & |[black]| 0 & 0 & 0 \\
                      0 & |[black]| 1 & 0 & |[black]| 0 & |[black]| 0 & 0 & 0 \\
                      0 & 0 & 0 & 0 & 0 & |[black]| 1 & 0 \\
                      0 & 0 & 0 & 0 & 0 & 0 & |[black]| 1 \\
            };
            
            \node[left=3pt of m-1-1] (left-0) {\small$\emptyseq$};
            \node[left=3pt of m-2-1] (left-1) {\small\seq{1}};
            \node[left=3pt of m-3-1] (left-2) {\small\seq{2}};
            \node[left=3pt of m-4-1] (left-3) {\small\seq{3}};
            \node[left=3pt of m-5-1] (left-4) {\small\seq{4}};
            \node[left=3pt of m-6-1] (left-5) {\small\seq{5}};
            \node[left=3pt of m-7-1] (left-6) {\small\seq{6}};
            
            \node[above=3pt of m-1-1] (top-0) {\small$\emptyseq$};
            \node[above=3pt of m-1-2] (top-1) {\small\seq{1}};
            \node[above=3pt of m-1-3] (top-2) {\small\seq{2}};
            \node[above=3pt of m-1-4] (top-3) {\small\seq{3}};
            \node[above=3pt of m-1-5] (top-4) {\small\seq{4}};
            \node[above=3pt of m-1-6] (top-5) {\small\seq{5}};
            \node[above=3pt of m-1-7] (top-6) {\small\seq{6}};
            
            \begin{pgfonlayer}{back}
            \fhighlight[black!25!white]{m-4-2}{m-5-2}
            \fhighlight[black!10!white]{m-4-4}{m-5-5}
            \end{pgfonlayer}
        \end{tikzpicture}
        \end{tabular}
\end{minipage}
        \caption{(Left) A simple sequential game with two players. Black round nodes are decision nodes of Player 1, white round nodes are decision nodes of Player 2. White square nodes represent terminal nodes. The set of sequences of Player 1 is $\Seqs[1]=\{\emptyseq,\seq{1},\ldots,\seq{6}\}$. (Right) Example of the matrix defining a coarse trigger deviation.
        }
        \label{fig:example_mapping}
\end{figure}

Let $\Ph\defeq\mleft\{\tdev:\hatinfo\in\Info,\hatpure[]\in\Pure_{\hatinfo}\mright\}$ be the set of all possible linear mappings defining coarse trigger deviation functions for Player $i$. 
Then, we define the concept of \emph{coarse trigger regret minimizer}. This notion will be shown to have a close connection with extensive-form coarse correlated equilibria.

\begin{definition}[Coarse Trigger Regret Minimizer]
For every $i\in[n]$, we call \emph{coarse trigger regret minimizer} for player i any $\Ph$-regret minimizer for the set of deterministic sequence-form strategies $\Pure$.
\end{definition}

\subsection{Extensive-Form Coarse Correlated Equilibria}

Equipped with the notion of coarse trigger deviation function we can provide the following definition of \emph{extensive-form coarse correlated equilibria} (EFCCE). 
\begin{definition}[$\epsilon$-EFCCE, EFCCE]\label{def:efcce}
    Given $\epsilon\ge 0$, a probability distribution $\muv\in\Delta^{|\Pi|}$ is an $\epsilon$-approximate EFCCE ($\epsilon$-EFCCE for short) if, by letting $\pure[]=(\pure[1],\ldots,\pure[n])\in\Pi$, for every player $i\in[n]$, and every coarse trigger deviation function $\tdev\in\Ph$, it holds 
	\begin{equation}\label{eq:efcce}
    \mathbb{E}_{\pure[]\sim\muv} \mleft[\ut\mleft(\tdev\mleft(\pure\mright),\pure[-i]\mright) - \ut\mleft(\pure[]\mright)\mright] \le \epsilon.
	\end{equation}
	A probability distribution $\muv\in\Delta^{|\Pi|}$ is an EFCCE if it is a 0-EFCCE.
\end{definition}
\noindent The above definition can be easily interpreted as the canonical definition of EFCCE by~\citet{farina2019coarse}. In their terminology, Equation~\eqref{eq:efcce} means that the expected utility of any trigger agent $(\hatinfo,\hatpure[])$ is never larger than the expected utility that Player $i$ would obtain by following recommendations by more than the amount $\epsilon$.

We can now prove one of the central results of the paper, which shows that if each player $i\in[n]$ in the game plays according to a $\Ph$-regret minimizer, then the empirical frequency of play approaches the set of EFCCEs (all omitted proofs are reported in Appendix~\ref{sec:appendix omitted proofs}).
\begin{restatable}{theorem}{efcceRegret}
\label{thm:euclid}
For each player $i\in[n]$, let $\puret[i][1],\ldots,\puret[i][T]$ be a sequence of deterministic sequence form strategies with cumulative $\Ph$-regret $\cumr$ with respect to the sequence of linear functions
\[
\ell^{(i),\,t}:\Pure\ni\pure\mapsto \ut\mleft(\pure,\puret[-i][t]\mright),
\]
and let $\muv\in\Delta^{|\Pi|}$ be the corresponding empirical frequency of play, defined as the probability distribution such that, for each $\pure[]=(\pure[1],\ldots,\pure[n])\in\Pi$, 
\[
\muv[\pure[]]\defeq \frac{1}{T}\sum_{t=1}^T\bbone[\pure[]=(\puret[1][t],\ldots,\puret[n][t])].
\]
Then, the empirical frequency of play $\muv$ is an $\epsilon$-EFCCE, with $\epsilon\defeq \frac{1}{T}\max_{i\in [n]}\cumr$.
\end{restatable}

\section{An Efficient No-Coarse-Trigger-Regret Algorithm}\label{sec:alg}

In this section we describe our efficient no-coarse-trigger regret minimizer.
Our approach follows the framework by \citet{Gordon08:No} (see \cref{sec:preliminaries regret minimization}). In order to apply the framework by~\citet{Gordon08:No} we need (i) to provide a regret minimizer for the set of coarse trigger deviation functions (\cref{sec:pt1}), and (ii) to show that for any $\phi\in\Ph$ it is possible to compute in poly-time a sequence-form strategy $\q\in\Q$ such that $\phi(\q)=\q$ (\cref{sec:closed form}). We will provide an efficient closed form solution for the latter problem by exploiting the particular structure of coarse trigger deviation functions.

\subsection{Regret Minimization for the Set $\co\Ph$}\label{sec:pt1}
Fix any player $i$. In this subsection, we discuss how one can construct an efficient regret minimizer for the \emph{convex hull} $\co\Ph$ of the set of coarse trigger deviation functions. Since $\co\Ph \supseteq \Ph$, any such regret minimizer for $\co\Ph$ is in particular a regret minimizer for $\Ph$. 
At a high level, our construction decomposes the problem of minimizing regret on $\Ph$ into $|\Info|$ subproblems, one for each possible trigger information set. Intuititevely, given any $\hat{I}\in\Info$, the subproblem for $\hat{I}$ corresponds to learning a continuation strategies for the trigger $\hat{I}$. In particular, each subproblem is itself a regret minimization problem, on the set of all continuation strategies $\q \in \Q_{\hat I}$, for which efficient regret minimization algorithms are known.

To make the above regret decomposition formal, we operate within the framework of \emph{regret circuits} \citep{Farina19:Regret}. Regret circuits provide ways to decompose the problem of minimizing regret over a composite set
into the problem of minimizing regret over the individual sets. Once regret minimizers for the individual sets have been constructed, an appropriate regret circuit will combine their outputs to guarantee low regret over the original composite set. For our purposes we will only need two regret circuit constructions: one for the convex hull, and one for affine transformations. We recall their main properties next.

\begin{proposition}\label{prop: convex hull}
Let $\cX_1,\ldots,\cX_m$ be a finite collection of sets, $\cR_1,\ldots,\cR_m$ be any regret minimizers for them, and $\cR_{\Delta^m}$ be a regret minimizers for the $m$-simplex. A regret minimizer $\cR_{\co}$ for the set $\co\{\cX_1,\ldots,\cX_m\}$ can be built as follows:
\begin{itemize}
    \item $\cR_{\co}.\textsc{NextElement}$ outputs the element $\lambda_1^tx_1^t+\ldots+\lambda_m^tx_m^t\in\co\{\cX_1,\ldots,\cX_m\}$, where, for each $j\in[m]$, $x_j^t$ is obtained by calling \textsc{NextElement} on $\cR_j$, and the vector $(\lambda_1^t,\ldots,\lambda_m^t)$ is obtained by calling \textsc{NextElement} on $\cR_{\Delta^m}$.
    \item $\cR_{\co}.\textsc{ObserveUtility}(L^t)$ forwards the linear utility function $L^t$ to each $\cR_j$, $j\in[m]$, and then calls $\cR_{\Delta^m}.\textsc{ObserveUtility}$ with the linear utility function $(\lambda_1,\ldots,\lambda_m)\mapsto L^t(x_1^t)\lambda_1+\ldots+L^t(x_m^t)\lambda_m$.
\end{itemize}
Furthermore, the composite regret minimizer $\cR_{\co}$ satisfies the regret bound $R_{\co}^T\le R^T_{\Delta^m}+\max_{j\in [m]} R_j^T$ for all $T$.
\end{proposition}

\begin{proposition}\label{prop: affine}
Let $\cX$ be a convex and compact set, $g:\cX\to\cY$ be an affine map, and $\cR_{\cX}$ any regret minimizer for the set $\cX$. Then, a regret minimizer $\cR_{g(\cX)}$ for the set $g(\cX)$ can be obtained as follows:
\begin{itemize}
    \item $\cR_{g(\cX)}.\textsc{NextElement}$ outputs $g(x^t)$, where $x^t$ is obtained by calling \textsc{NextElement} on $\cR_{\cX}$.
    \item $\cR_{g(\cX)}.\textsc{ObserveUtility}(L^t)$ forwards the linear utility $x\mapsto L^t(g(x))-L^t(g(0))$ to the $\cR_{\cX}$ regret minimizer.
\end{itemize}
Furthermore, the composite regret minimizer $\cR_{g(\cX)}$ satisfies the regret bound $R^T_{g(\cX)} = R^T_{\cX}$ at all times $T$.
\end{proposition}

We start from the following observation:
\begin{align*}
    \co\Ph
    &= \co\mleft\{\bar{\Psi}^{(i)}_{\hatinfo}:\hatinfo\in\Info\mright\}, \text{ where}
\\
\bar{\Psi}^{(i)}_{\hatinfo}&\defeq\co\mleft\{\tdev[\hatinfo][\hatpure[]]:\hatpure[]\in\Pure_{\hatinfo}\mright\}.
\end{align*}
Hence, by virtue of the convex hull regret circuit (\cref{prop: convex hull}), in order to construct a regret minimizer for $\co\Ph$ it suffices to construct regret minimizers for each of the ${\bar\Psi}^{(i)}_{\hatinfo}$. Consider now the mapping $g_{\hatinfo}:\bbR^{|\Seqs_{\hatinfo}|}\ni\vec{x}\mapsto \tdev[\hatinfo][\vec{x}]$, which is promptly verified to be affine by definition of coarse trigger deviation function (see~\cref{sec: deviation functions}). Then,
\begin{align*}
{\bar\Psi}^{(i)}_{\hatinfo} = 
g_{\hatinfo}\mleft(\co\Pure_{\hatinfo}\mright)
= 
g_{\hatinfo}\mleft(\Q_{\hatinfo}\mright) \qquad \forall\ \hatinfo\in\Info.
\end{align*}
In other words, the set ${\bar\Psi}^{(i)}_{\hatinfo}$ is the image of $\Q_{\hatinfo}$ under an affine transformation. By using the affine transformation regret circuit (\cref{prop: affine}), a regret minimizer for ${\bar\Psi}^{(i)}_{\hatinfo}$ can be constructed from any regret minimizer for $\Q_{\hatinfo}$. Because efficient regret minimizers for $\Q_{\hatinfo}$ are well known in the literature (e.g., the CFR algorithm by~\citet{zinkevich2008regret}), our approach based on regret circuits yields an efficient regret minimizer for $\co\Ph$, summarized in \cref{alg:regret minimizer}.
\begin{restatable}{theorem}{thRegretMinimizer}
\label{th: regret minimizer}
Let $\cR$ be a regret minimizer instantiated as specified in~\cref{alg:regret minimizer}, and employing the CFR algorithm for each $\cR_{\hatinfo}$ and the regret matching algorithm for $\cR_{\Delta^{|\Info|}}$. After observing a sequence of linear utility functions $L^1,\ldots,L^T:\co\Ph\to\bbR$ with range upper bounded by $U$ (that is, for all $t\in T$ $\max_{\phi,\phi'\in\co\Ph}\{L^t(\phi)-L^t(\phi')\}\le U$), the regret cumulated by the transformations $\phi^1,\ldots,\phi^T\in\co\Ph$ produced by $\cR$ is such that \[
R^T=\max_{\phi^\ast\in\co\Ph}\sum_{t=1}^T\mleft(L^t(\phi^\ast)-L^t(\phi^t)\mright)\le 2 U |\Sigma^{(i)}|\sqrt{T}.
\]
\end{restatable}

\begin{algorithm}[tb]
\caption{Regret minimizer for the set $\co\Ph$}
\label{alg:regret minimizer}
\textbf{Data}: player $i\in [n]$, one regret minimizer $\cR_{\hatinfo}$ for the set $\bar{\Psi}^{(i)}_{\hatinfo}$ for each $\hatinfo\in\Info$, one regret minimizer $\cR_{\Delta^{|\Info|}}$ for the $|\Info|$-simplex. 

\hrulefill

\textbf{function} $\textsc{NextElement}()$
\begin{algorithmic}[1] 
\STATE $\lambdav^t\gets \cR_{\Delta^{|\Info|}}.\textsc{NextElement}()$
\STATE \textbf{for} $\hatinfo\in\Info$ \textbf{do} $\tdev[\hatinfo][\q^t_{\hatinfo}]\gets  \cR_{\hatinfo}.\textsc{NextElement}()$
\STATE \textbf{return} $\sum_{\hatinfo\in\Info} \lambdav^t[\hatinfo]\tdev[\hatinfo][\q^t_{\hatinfo}]$
\end{algorithmic}

\hrulefill

\textbf{function} $\textsc{ObserveUtility}(L^t)$
\begin{algorithmic}[1]
    \STATE \textbf{for} $\hatinfo\in\Info$ \textbf{do} $ \cR_{\hatinfo}.\textsc{ObserveUtility}(L^t)$
\STATE $\ell^t_{\lambdav}\gets \Delta^{|\Info|}\ni \lambdav\mapsto \sum_{\hatinfo\in\Info}\lambdav[\hatinfo]L^t(\tdev[\hatinfo][\q^t_{\hatinfo}])$\label{line:loss rm delta}
\STATE $ \cR_{\Delta^{|\Info|}}.\textsc{ObserveUtility}(\ell^t_{\lambdav})$
\end{algorithmic}
\end{algorithm}

\subsection{Closed-Form Fixed Point Computation}\label{sec:closed form}

\begin{algorithm}[tb]
\caption{\textsc{FixedPoint}$(\phi)$}
\label{alg:fixed point}
\textbf{Input}: $\phi=\sum_{\hatinfo\in\Info}\lambdav[\hatinfo]\tdev[\hatinfo][\q_{\hatinfo}]\in\co\Ph$\\
\textbf{Output}: $\q\in\Q$ fixed point of $\phi$
\begin{algorithmic}[1] 
\STATE $\q^\star \gets \vec{0} \in\bbR^{|\Seqs|}_{\ge0}, \quad \q^\star[\emptyseq] \gets 1$\label{line:q star init}
\FOR{$\sigma=(\info,a)\in\Seqs\setminus\{\emptyseq\}$ in top-down ($\prec$) order}\label{line:for loop}
\STATE\label{eq:denominator fixed point} 
$d_{\sigma} \gets \sum_{\substack{\info'\in\Info:\info'\preceq \info}} \lambdav[\info']\label{line:dsigma}
$
\STATE \textbf{if} {$d_{\sigma}=0$} \textbf{then} 
 $\displaystyle\q^\star[\sigma]\gets \frac{\q^\star[\parseq[i][\info]] }{ |\A(I)| }$\label{line: set to uniform}
\STATE \textbf{else} \label{eq: fixed point}$
\displaystyle\q^\star[\sigma]\gets \frac{1}{d_{\sigma}}\mleft(\sum_{\info'\in\Info:\info' \preceq \info}\!\!\!\!\!\lambdav[\info']\,\q_{\info'}[\sigma]\,\q^\star[\parseq[i][\info']]\mright)
$
\ENDFOR
\STATE \textbf{return} $\q^\star$
\end{algorithmic}
\end{algorithm}

\begin{algorithm}[tb]
\caption{Regret minimizer for the set $\co\Ph$}
\label{algo:final}
\textbf{Data}: player $i\in [n]$, $\cR$ regret minimizer for the set $\co\Ph$ defined in \cref{alg:regret minimizer}. 

\hrulefill

\textbf{function} $\textsc{NextStrategy}()$
\begin{algorithmic}[1] 
\STATE $\phi^t=\sum_{\hatinfo\in\Info}\lambdav^t[\hatinfo]\,\tdev[\hatinfo][\q^t_{\hatinfo}] \gets \cR.\textsc{NextElement}()$
\STATE \textbf{return} $x^t \gets \textsc{FixedPoint}(\phi^t)$
\end{algorithmic}

\hrulefill

\textbf{function} $\textsc{ObserveUtility}(\ell^t)$
\begin{algorithmic}[1]
\STATE Define the linear function $L^t: \phi \mapsto \ell^t(\phi(x^t))$
\STATE $\cR.\textsc{ObserveUtility}(L^t)$
\end{algorithmic}
\end{algorithm}

\noindent We complete the construction of our $\Ph$-regret minimizer by showing that, for any player $i\in[n]$ and $\phi\in\co\Ph$, it is possible to compute a sequence-form strategy $\q^\star\in\Q$ such that $\phi(\q^\star)=\q^\star$ in linear time in $|\Seqs| D^{(i)}$, where $D^{(i)}$ denotes the maximum depth of the game tree considering only actions of Player $i$ (i.e., $D^{(i)}=\max_{z\in\Z}|\{\sigma\in\Seqs:\sigma\prec\parseq[i][z]\}|$).
Unlike all known results for EFCE, our result does not rely---as an intermediate step---on the computation the stationary distributions of some stochastic matrices~\cite{farina2021simple,morrill2021efficient}, and does not require to manage any internal-regret minimizer as in~\citet{Celli20:NoRegret}.
Instead, we show that the fixed-point strategy $\q^\star$ can be found in closed form for any $\phi \in \co\Ph$.

Let $\phi=\sum_{\hatinfo\in\Info}\lambdav[\hatinfo]\tdev[\hatinfo][\q_{\hatinfo}]$ be any deviation function returned by \cref{alg:regret minimizer}, where $\lambdav\in\Delta^{|\Info|}$, and $\q_{\hatinfo}\in\Q_{\hatinfo}$ for each $\hatinfo\in\Info$.
%
%
Equipped with this additional notation, \cref{alg:fixed point} describes an efficient procedure to compute a fixed point of a given transformation $\phi\in\co\Ph$. 
The algorithm iterates over sequences of Player $i$ according to their partial ordering $\prec$ (i.e., it is never the case that a sequence $\sigma=(I,a)$ is considered before $\parseq[i][I]$). For each sequence $\sigma=(I,a)$, the algorithm computes $d_{\sigma}\in\bbR_{\ge0}$ as the sum of the weights of the convex combination corresponding to information sets preceding $I$ (Line~\ref{eq:denominator fixed point}). If $d_{\sigma}=0$, then the matrix $M$ corresponding to the transformation $\phi$ must be such that $M[\sigma,\sigma]=1$, and $M[\sigma,\sigma']=0$ for all $\sigma'\ne\sigma$. Therefore, the choice we make at $\sigma$ is indifferent as long as the resulting $\q^\star$ is a well-formed sequence-form strategy. We set $\q^\star$ so that the probability-mass flow is evenly divided among sequences originating in $I$ (Line~\ref{line: set to uniform}).
Finally, when $d_{\sigma}\ne 0$, Line~\ref{eq: fixed point} assigns to $\q^\star[\sigma]$ a value equal to a weighted sum of $\q_{\info'}[\sigma]\q^\star[\sigma']$ for sequences $\sigma'=(I',a')$ preceding information set $I$.
The following theorem shows that ~\cref{alg:fixed point} is indeed correct and behaves as desired. 
\begin{restatable}{theorem}{closedForm}
\label{thm:closedForm}
For any player $i\in[n]$, and transformation $\phi=\sum_{\hatinfo\in\Info}\lambdav[\hatinfo]\tdev[\hatinfo][\q_{\hatinfo}]\in\co\Ph$, the vector $\q^\star\in\bbR^{|\Seqs|}$ obtained through~\cref{alg:fixed point} is such that $\q^\star\in\Q$, and $\phi(\q^\star)=\q^\star$.~\cref{alg:fixed point} runs in linear time in $|\Seqs|D^{(i)}$.
\end{restatable}


As a direct consequence of the correctness (\cref{th: regret minimizer} and \cref{thm:closedForm}) of the two steps required by \citet{Gordon08:No}'s construction (recalled in \cref{sec:preliminaries regret minimization}) we have the following.
\begin{restatable}{corollary}{corollaryEfcce}
\cref{algo:final} is a $\co\Ph$-regret minimizer for the set $\Q$. Thus, when all player play according to \cref{algo:final} where at all $t$ the utility $\ell^t$ of each player is set to their own linear utility function given the opponents' actions, the empirical frequency of play in the game after $T$ iterations converges to a $O(1/\sqrt{T})$-EFCCE with high probability, and converges almost surely to an EFCCE in the limit.  
\end{restatable}

\section{Experimental Evaluation}

We experimentally investigate the convergence of our no-regret learning dynamics for EFCCE on four
standard benchmark games: 3-player Kuhn poker, 3-player Leduc poker, 3-player Goofspiel, 2-player (general-sum) Battleship. See \cref{sec:exp details} for a description of each game. 
For each game we investigate the rate of convergence to EFCCE, measured as the maximum expected increase in utility that any player can obtain by optimally responding to the mediator that recommends strategies based on the correlated distribution of play induced by the learning dynamics, of our EFCCE learning dynamics and of the very related EFCE learning dynamics in~\citet{farina2021simple}, which were obtained using the same framework as this paper. 
Experimental results are available in \cref{fig:experiments}. In each game, we ran our EFCCE dynamics and the EFCE dynamics for the same number of iterations (1000 iterations in Leduc poker and Goofspiel, 5000 in Kuhn poker, 10000 in Battleship). Each fixed-point computation in the EFCE dynamics was performed through an optimized implementation of the power iteration method. The power iteration was interrupted when the Euclidean norm of the fixed-point residual got below the threshold value of $10^{-6}$ (this typically was achieved in less than 10 iterations of the method).

In all four games, we observe that our EFCCE dynamics outperform the EFCE dynamics, often by a significant margin. This is consistent with intuition, as the EFCE dynamics are solving a strictly harder problem (minimizing the EFCE gap, instead of the EFCCE gap). In Goofspiel, the EFCCE dynamics 
induce a correlated distribution of play that is an exact EFCCE after roughly 500 iterations.

In Kuhn, Leduc, and Goofspiel the runtime of each iteration is comparable between EFCE and EFCCE dynamics, while in the Battleship game, the EFCCE dyanmics are roughly 30\% faster per iteration. This is consistent with the observation that the amout of work necessary to find a stationary distribution grows approximately cubically with the maximum number of actions at any decision point in the game, which is higher in Battleship compared to the other games.

\begin{figure}
    \centering
    \includegraphics[scale=.66]{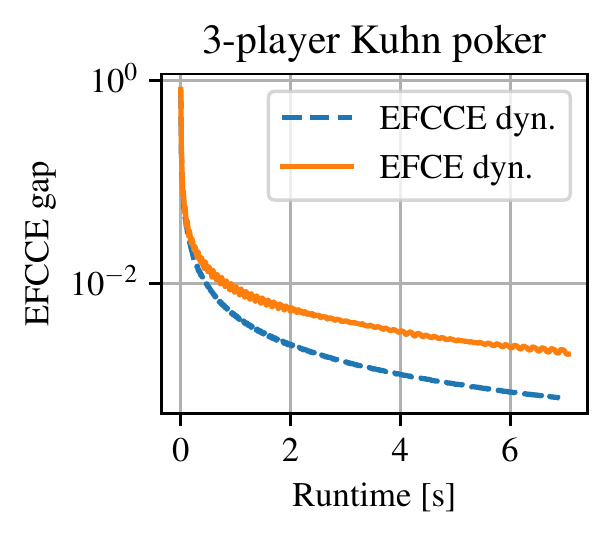}
    \includegraphics[scale=.66]{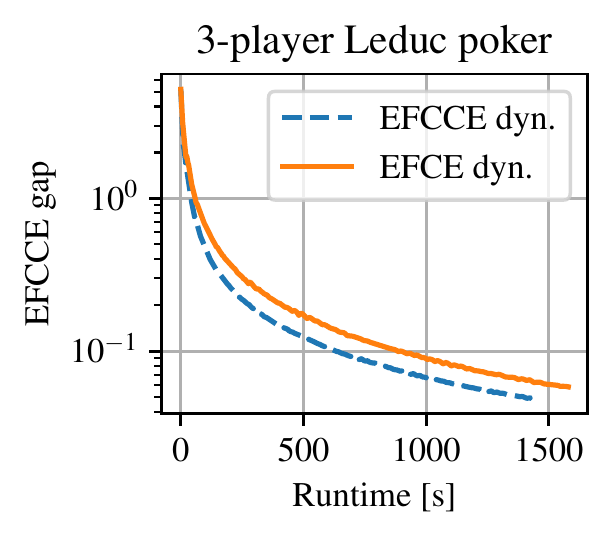}
    \includegraphics[scale=.66]{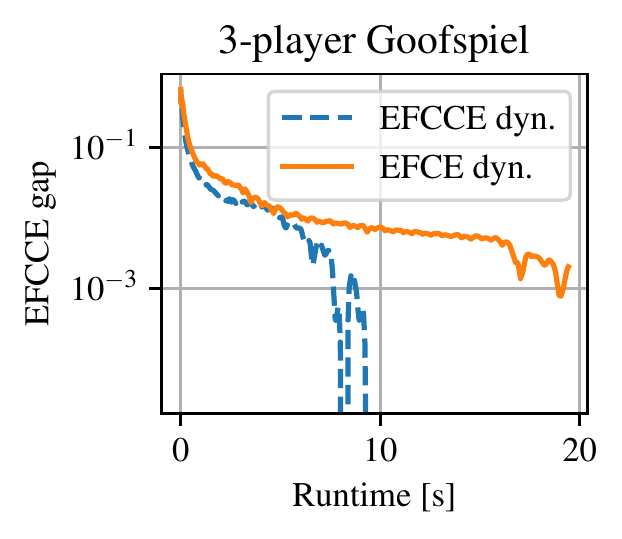}
    \includegraphics[scale=.66]{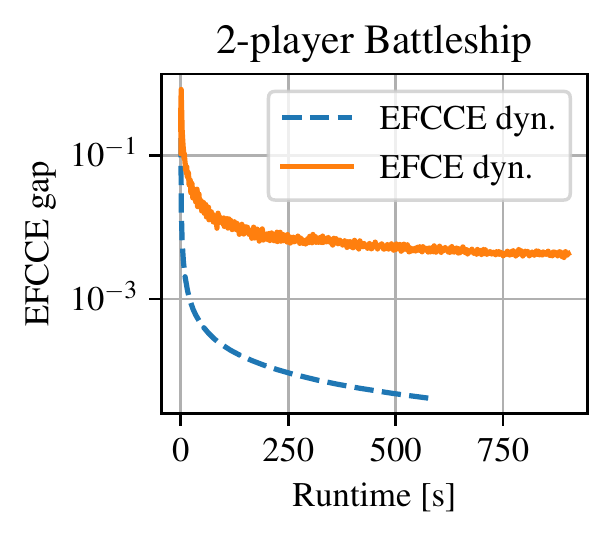}
    \caption{Rate of convergence to EFCCE (measured through the EFCCE gap) of the EFCCE learning dynamics introduced in this paper and of the related EFCE learning dynamics by \citet{farina2021simple} on four standard benchmark games. The y-axis is on a log scale.}
    \label{fig:experiments}\vspace{-0mm}
\end{figure}

\section{Conclusions and Future Research}

We showed that---at least when analyzed through the phi-regret minimization framework---the computation of the fixed point required at each iteration in EFCCE learning dynamics can be carried out in closed form using a simple formula that avoids the computation of stationary distributions of Markov chains. This is contrast to all known learning dynamics for EFCE.

We conjecture that a similar result could be derived within the EFR framework, when blind causal deviations are considered, though we leave exploration of that direction open.

\clearpage
\section*{Acknowledgments}
This material is based on work supported by the National
Science Foundation under grants IIS-1718457, IIS-1901403, and CCF-1733556, and the ARO under award W911NF2010081. Gabriele Farina
is supported by a Facebook fellowship.

\bibliography{aaai22}

\iftrue
\clearpage
\appendix

\onecolumn
%
\section{Proofs}\label{sec:appendix omitted proofs}

\subsection{Proof for \cref{sec:efcce}}

\efcceRegret*
\begin{proof}
    For any $i\in [n]$, we have $\cumr\le\epsilon T$. Then, by definition of $\cumr$ as per Equation~\eqref{eq:cum phi regret}, for any coarse trigger deviation function $\phi\in\Ph$ is must hold
    \begin{align*}
    T\epsilon & \ge \sum_{t=1}^T\mleft(\ell^{(i),\, t}\mleft(\phi(\puret)\mright)-\ell^{(i),\, t}\mleft(\puret\mright)\mright)
    =\sum_{t=1}^T\mleft(\ut\mleft(\phi(\puret),\puret[-i][t]\mright) - \ut\mleft(\pure[]^t\mright)\mright).
    \end{align*}
    This yields the following
    \begin{align*}
        T\epsilon &\ge \sum_{t=1}^T\sum_{\pure[]\in\Pi}\bbone[\pure[]=\pure[]^t] \mleft(\ut\mleft(\phi(\pure),\pure[-i]\mright) - \ut\mleft(\pure[]\mright)\mright) \\
        & = \sum_{\pure[]\in\Pi}\mleft(\sum_{t=1}^T\bbone[\pure[]=\pure[]^t]\mright) \mleft(\ut\mleft(\phi(\pure),\pure[-i]\mright) - \ut\mleft(\pure[]\mright)\mright) \\
        & = T \sum_{\pure[]\in\Pi}\muv[\pure[]]\mleft(\ut\mleft(\phi(\pure),\pure[-i]\mright) - \ut\mleft(\pure[]\mright)\mright).
    \end{align*}
    This is precisely the definition of $\epsilon$-EFCCE (Definition~\ref{def:efcce}), as we wanted to show.
\end{proof}

\subsection{Proof for \cref{sec:pt1}}

\thRegretMinimizer*
\begin{proof}
We start by recalling the known regret bounds for CFR~\citep[Theorems 3 and 4]{zinkevich2008regret} and regret matching~\citep{hart2000simple}.

\begin{lemma}[Known regret bound for CFR]\label{lemma:cfr}
Let $i\in[n]$ and $\info\in\Info$, and consider the CFR algorithm run on the set $\Q_{\info}$. Then, for any sequence of linear utility functions $\tilde{\ell}^1,\ldots,\tilde{\ell}^T : \Q \to \bbR$ with range upper bounded as $\max_{\q,\q' \in \Q_{\info}}(\tilde{\ell}^t(\q) - \tilde{\ell}^t(\q')) \le U$ at all $t$, the regret cumulated by the CFR algorithm satisfies the inequality
\[
        R_{\textsc{cfr}}^T\le U|\Seqs_{\info}|\sqrt{T}.
\]
\end{lemma}
\begin{lemma}[Known regret bound for regret matching]\label{lemma:regret matching}
Consider the regret matching algorithm applied to a simplex domain $\Delta^m$. For any sequence of linear utility functions $\tilde{\ell}^1,\ldots,\tilde{\ell}^T : \Delta^m \to \bbR$ with range upper bounded as $\max_{\vec{x},\vec{x}' \in \Delta^m}(\tilde{\ell}^t(\vec{x}) - \tilde{\ell}^t(\vec{x}')) \le U$ at all $t$, the regret cumulated by the regret matching algorithm satisfies the inequality
    \[
        R_{\textsc{rm}}^T\le U m \sqrt{T}.
    \]
\end{lemma}
\noindent
By~\cref{prop: convex hull} and by construction of Algorithm~\ref{th: regret minimizer} we have that
\begin{equation}\label{eq:regret bound}
R^T\le R_{\Delta^{|\Info|}}^T + \max_{\hatinfo\in\Info} R_{\hatinfo}^T.
\end{equation}
The loss function observed by the regret minimizer $\cR_{\Delta^m}$ at time $t$ is $\ell^t_{\lambdav}\gets \Delta^{|\Info|}\ni \lambdav\mapsto \sum_{\hatinfo\in\Info}\lambdav[\hatinfo]L^t(\tdev[\hatinfo][\q^t_{\hatinfo}])$ (Line~\ref{line:loss rm delta}). Since $\tdev(\q)\in\Q$ for any $\q\in\Q$, we have that the maximum range of $\ell_{\lambdav}^t$ is at most equal to the maximum range of $L^t$. Therefore, by~\cref{lemma:regret matching} we have $R_{\Delta^{|\Info|}}^T\le U|\Info|\sqrt{T}$. 
We not turn our attention to the regret minimizers $\cR_{\hatinfo}$. Fix $\hatinfo\in\Info$, by Proposition~\ref{prop: affine} the loss observed at time $t$ by the CFR algorithm running on the set $Q_{\hatinfo}^T$ is $L^t(g_{\hatinfo}(\cdot))-L^t(g_{\hatinfo}(0))$. Then, the range of this linear function is equal to $\max_{\q,\q'\in\Q_{\hatinfo}} L^t(g_{\hatinfo}(\q))-L^t(g_{\hatinfo}(\q'))$, which is upper bounded by $U$ since $g_{\hatinfo}$ maps sequence-form strategies into valid sequence-form strategies. By Proposition~\ref{prop: affine} we have that $R_{\hatinfo}^T$ is at most equal to the regret cumulated by CFR run on $\Q_{\hatinfo}$. This, together with  \cref{lemma:cfr}, yields that, for any $\hatinfo\in\Info$, $R^T_{\hatinfo}\le U|\Seqs_{\hatinfo}|\sqrt{T}$.
Then, by substituting into~\eqref{eq:regret bound},
\[
R^T\le U|\Info|\sqrt{T} + \max_{\hatinfo\in\Info} U|\Seqs_{\hatinfo}|\sqrt{T}\le U|\Info|\sqrt{T} + U|\Seqs|\sqrt{T} \le 2|\Seqs|\sqrt{T}.
\]
This concludes the proof.
\end{proof}

\corollaryEfcce*
\begin{proof}
    \cref{th: regret minimizer} establishes that \cref{alg:regret minimizer} is a regret minimizer for the set $\co\Ph \supseteq \Ph$. \cref{thm:closedForm} establishes that \cref{alg:fixed point} returns a fixed point $\q \in \Q$ for any $\phi \in \co\Ph$. Hence, by using the result by \citet{Gordon08:No}, \cref{algo:final} is a $\Ph$-regret minimizer for the set $\Q$ for each player $i$. At each time $t$, \cref{algo:final} returns a randomized strategy $\q^{(i),t}\in\Q$ that Player~$i$ should play. A standard application of the Azuma-Hoeffding inequality shows that by sampling actions according to $\q^{(i),t}$, the $\Ph$-regret incurred by Player~$i$ grows by an amount bounded above by $O(\sqrt{T\log(1/\delta)})$ with probability at least $1-\delta$, for any $\delta\in(0,1)$. Hence, by invoking \cref{thm:euclid}, with probability at least $1-\delta$, after any $T$ iterations the empirical frequency of play is an $\epsilon$-EFCCE where
    \begin{equation}\label{eq:delta bound}
        \epsilon = O\mleft(\frac{1}{T}(\sqrt{T} + \sqrt{T\log(1/\delta)})\mright) = O\mleft(\frac{1}{\sqrt{T}}(1 + \log(1/\delta))\mright).
    \end{equation}
    This concludes the proof of the first part of the statement. Going from the high-probability regret guarantee for any $\delta \in (0,1)$ and $T$ given in \eqref{eq:delta bound} to almost-sure convergence in the limit as $T\to\infty$ is a direct application of the classic Borel-Cantelli lemma.
\end{proof}

\subsection{Proof for \cref{sec:closed form}}

The following result will be useful when proving~\cref{thm:closedForm}.
\begin{lemma}\label{lemma: phi decomposition}
For any $\phi=\sum_{\info'\in\Info}\lambdav[\info']\tdev[\info'][\q_{\info'}]\in\co\Ph$, $\q\in\Q$, and $\sigma=(I,a)\in\Seqs$, it holds
\[
\phi(\q)[\sigma]-\q[\sigma]=\mleft(\sum_{\info'\preceq \info}\lambdav[\info']\q_{\info'}[\sigma]\q[\parseq[i][\info']]\mright)-d_{\sigma}\q[\sigma].
\]
\end{lemma}
\begin{proof}
Fix a sequence $\sigma\in\Seqs$. Then, by definition of the linear mappings $\tdev[\info'][\q_{\info'}]$, we have
\begin{align*}
    \phi(\q)[\sigma] & = \sum_{\info'\in\Info}\lambdav[\info']\tdev[\info'][\q_{\info'}](\q)[\sigma]
    \\ & = 
    \sum_{\info'\in\Info}\lambdav[\info']\begin{cases} \q_{\info'}[\sigma] \q[\parseq[i][\info']] & \text{if }\sigma\succeq\info'\\ \q[\sigma]  & \text{\normalfont otherwise}\end{cases}
    \\ & = 
    \mleft(1-\sum_{\info'\preceq\sigma}\lambdav[\info']\mright)\q[\sigma] 
    + \sum_{\info'\preceq \sigma}\lambdav[\info']\q_{\info'}[\sigma]\q[\parseq[i][\info']].
\end{align*}
By re-arranging the above equation we obtain the statement.
\end{proof}

\closedForm*
\begin{proof}
The proof is divided into three parts: (i) we show that, for any $\phi\in\co\Ph$, the vector $\q^\star\in\bbR^{|\Seqs|}$ obtained through Algorithm~\ref{alg:fixed point} is such that $\q^\star\in\Q$ (i.e., it is a valid sequence-form strategy); (ii) we show that, for any $\phi\in\co\Ph$, the sequence-form strategy $\q^\star$ obtained via Algorithm~\ref{alg:fixed point} is such that $\phi(\q^\star)=\q^\star$; (iii) finally, we show that Algorithm~\ref{alg:fixed point} runs in time $O(|\Seqs|D^{(i)})$.

\textit{Part 1: $\q^\star$ is a sequence-form strategy}. By construction (Line~\ref{line:q star init}), $\q^\star[\emptyseq]=1$. Then, we need to show that, for each $\info\in\Info$, $\sum_{a\in\A(\info)}\q^\star[(I,a)]=\q^\star[\parseq[i][\info]]$ (see \cref{def:seq_polytope}). For any $\info\in\Info$ such that $d_{\sigma}=0$ it is immediate to see that the above constraint holds by construction (Line~\ref{line: set to uniform}). For each $\info\in\Info$ such that $d_{\sigma}\ne 0$ we have that
\begin{align*}
    \sum_{a\in\A(\info)}\q^\star[(\info,a)]  &=
    \frac{1}{d_{\sigma}}\mleft(\sum_{a\in\A(\info)}\sum_{\info'\preceq\info}\lambdav[\info']\q_{\info'}[(I,a)]\q^\star[\parseq[i][\info']]\mright)
    \\ &  =
    \frac{1}{d_{\sigma}}\mleft(\sum_{\info'\preceq\info}\lambdav[\info']\q^\star[\parseq[i][\info']]\mleft(\sum_{a\in\A(\info)}\q_{\info'}[(I,a)]\mright)\mright)
    \\ &  =
    \frac{1}{d_{\sigma}}\mleft(\sum_{\info'\preceq\info}\lambdav[\info']\q^\star[\parseq[i][\info']] \cdot\begin{cases} \q_{\info'}[\parseq[i][\info]] & \text{if } I' \prec I\\ 1 & \text{otherwise}\end{cases}\mright),
\end{align*}
where the first equality holds by Line~\ref{eq: fixed point} in Algorithm~\ref{alg:fixed point}, and the last equality holds because $\q_{\info'}\in \Q_{\info'}$.
%
We distinguish between two cases: if $d_{\parseq[i][\info]}=0$, then $\lambdav[\info']=0$ for each $\info'\prec\info$. Therefore, since we are assuming $d_{\sigma}\ne 0$, it must be the case that $d_{\sigma}=\lambdav[\info]\ne 0$. This yields the following
\begin{align*}
\sum_{a\in\A(\info)}\q^\star[(\info,a)] & =
\frac{1}{d_{\sigma}}\mleft(\sum_{\info'\preceq\info}\lambdav[\info']\q^\star[\parseq[i][\info']]\cdot\begin{cases} \q_{\info'}[\parseq[i][\info]] & \text{if } I' \prec I\\ 1 & \text{otherwise}\end{cases}\mright)
\\ & = \frac{1}{\lambdav[\info]}\mleft(\lambdav[\info]\q^\star[\parseq[i][\info]]\mright)
 = \q^\star[\parseq[i][\info]].
\end{align*}
Contrarily, if $d_{\parseq[i][\info]}\ne 0$, then $\q^\star[\parseq[i][\info]]$ was set according to Line~\ref{eq: fixed point}, and thus 
\begin{equation}\label{eq:parseq fixed point}
\q^\star[\parseq[i][\info]]=\frac{1}{d_{\parseq[i][\info]}}\mleft(\sum_{\info'\prec \info}\lambdav[\info']\q^\star[\parseq[i][\info']]\q_{\info'}[\parseq[i][\info]]\mright).
\end{equation}
By definition of $d_\sigma$ (Line~\ref{line:dsigma}), $d_{\sigma}=d_{\parseq[i][\info]}+\lambdav[\info]$. Then,
\begin{align*}
    \sum_{a\in\A(\info)}\q^\star[(\info,a)]  & 
     =
    \frac{1}{d_{\sigma}}\mleft(\sum_{\info'\preceq\info}\lambdav[\info']\q^\star[\parseq[i][\info']]\cdot\begin{cases} \q_{\info'}[\parseq[i][\info]] & \text{if } I' \prec I\\ 1 & \text{otherwise}\end{cases}\mright)
    \\ &  =
    \frac{1}{d_{\parseq[i][\info]}+\lambdav[\info]}\Bigg(\lambdav[\info]\q^\star[\parseq[i][\info]]
    +\sum_{\info'\prec\info}\lambdav[\info']\q^\star[\parseq[i][\info']]\q_{\info'}[\parseq[i][\info]]\Bigg)
    \\ &  =
    \frac{1}{d_{\parseq[i][\info]}+\lambdav[\info]}\mleft(\lambdav[\info]\q^\star[\parseq[i][\info]] + d_{\parseq[i][\info]} \q^\star[\parseq[i][\info]]\mright)
    =
    \q^\star[\parseq[i][\info]],
\end{align*}
where the second to last equality is obtained by Equation~\eqref{eq:parseq fixed point}. This concludes the first part of the proof.

\textit{Part 2: $\q^\star$ is a fixed point of $\phi$.} Fix a sequence $\sigma=(I,a)\in\Seqs$. We want to show that $\phi(\q^\star)[\sigma]-\q^\star[\sigma]=0$. If $\sum_{\info'\preceq\info}\lambdav[\info']=0$, then it immediately holds that $\phi(\q^\star)[\sigma]=\q^\star[\sigma]$. Otherwise, if $\sum_{\info'\preceq\info}\lambdav[\info']\ne 0$, by applying \cref{lemma: phi decomposition} and by subsequently substituting $\q^\star[\sigma]$ according to Line~\ref{eq: fixed point}, we obtain
\begin{align*}
    \phi(\q^\star)[\sigma]-\q^\star[\sigma] & 
    = 
    \mleft(\sum_{\info'\preceq \info}\lambdav[\info']\q_{\info'}[\sigma]\q^\star[\parseq[i][\info']]\mright)-d_{\sigma}\q^\star[\sigma]
    \\ &  =  
    \mleft(\sum_{\info'\preceq \info}\lambdav[\info']\q_{\info'}[\sigma]\q^\star[\parseq[i][\info']]\mright)
    - 
    \frac{d_{\sigma}}{d_{\sigma}}\mleft(\sum_{\info'\preceq \info}\lambdav[\info']\q_{\info'}[\sigma]\q^\star[\parseq[i][\info']]\mright)
    = 0.
\end{align*}
This concludes this part of the proof.

\textit{Part 3: time complexity.} For each sequence in $\Seqs\setminus\{\emptyseq\}$ (Line~\ref{line:for loop}), \cref{alg:fixed point} has to visit at most $D^{(i)}$ information sets as part of the sums required on Lines~\ref{line:dsigma} and~\ref{eq: fixed point}. This completes the proof.
\end{proof}

\section{Experimental Evaluation}\label{sec:exp details}

\subsection{Description of game instances}

The size (in terms on number of infosets and sequences) of the parametric instances we use as benchmark is described in~\cref{fig:dimensions}.
In the following, we provide a detailed explanation of the rules of the games.

\begin{figure}[H]
	\centering
		{\begin{minipage}{.4\textwidth}
			\small
			\begin{tabular}{c@{\hskip 6pt}ccccc}
				& \textbf{Players} & \textbf{Ranks} & \textbf{Player}& \textbf{Infosets} & \textbf{Sequences} \\[-.1mm]
				\midrule\\[-3mm]
				\multirow{3}{*}{Kuhn} & \multirow{3}{*}{3} & \multirow{3}{*}{4} & Player 1 & 16 & 33  \\[-.1mm]
				&&&  Player 2 & 16 & 33  \\[-.1mm]
				&&&  Player 3 & 16 & 33  \\[-.1mm]
				\midrule
				\multirow{3}{*}{Goofspiel} & \multirow{3}{*}{3} & \multirow{3}{*}{3} & Player 1 & 837 &  934 \\[-.1mm]
				&&&  Player 2 & 837 & 934  \\[-.2mm]
				&&&  Player 3 & 837 & 934  \\[-.2mm]
				\midrule
				\multirow{3}{*}{Leduc} & \multirow{3}{*}{3} & \multirow{3}{*}{3} & Player 1 & 3294 &  7687 \\[-.1mm]
				&&& Player 2 & 3294 & 7687  \\[-.1mm]
				&&& Player 3 & 3294 & 7687  \\[-.1mm]
				\bottomrule
			\end{tabular}\\[2mm]
		
			\hspace{-.5mm}
			\begin{tabular}{c@{\hskip 6pt}ccccc}
			& \textbf{Grid} & \textbf{Rounds} & \textbf{Player} & \textbf{Infosets} & \textbf{Sequences} \\[-.1mm]
			\midrule\\[-3mm]
			\multirow{2}{*}{Battleship} & \multirow{2}{*}{$(3,2)$} & \multirow{2}{*}{$3$} & Player 1 & 1413 & 2965 \\[-.1mm]
			&  &  & Player 2 &  1873 & 4101 \\[-.1mm]
			\bottomrule
		\end{tabular}
		\end{minipage}}
		\caption{Size of our game instances in terms of number of sequences and infosets for each player of the game.}
		\label{fig:dimensions}
\end{figure}

\paragraph{Kuhn poker}
The two-player version of the game was originally proposed by~\cite{kuhn1950simplified}, while the three-player variation is due to~\cite{farina2018exante}.
In a three-player Kuhn poker game with rank $r$, there are $r$ possible cards. Each player initially pays one chip to the pot, and she/he is dealt a single private card. 
The first player may {\em check} or {\em bet} ({\em i.e.,} put an additional chip in the pot). Then, the second player can check or bet after a first player's check, or {\em fold/call} the first player's bet. If no bet was previously made, the third player can either check or bet. Otherwise, she/he has to fold or call. After a bet of the second player (resp., third player), the first player (resp., the first and the second players) still has to decide whether to fold or to call the bet. At the showdown, the player with the highest card who has not folded wins all the chips in the pot.

\paragraph{Goofspiel}
This game was originally introduced by~\cite{ross1971goofspiel}. Goofspiel is essentially a bidding game where each player has a hand of cards numbered from $1$ to $r$ ({\em i.e.}, the rank of the game). A third stack of $r$ cards is shuffled and singled out as prizes. 
Each turn, a prize card is revealed, and each player privately chooses one of her/his cards to bid, with the highest card winning the current prize. In case of a tie, the prize card is discarded. After $r$ turns, all the prizes have been dealt out and the payoff of each player is computed as follows: each prize card’s value is equal to its face value and the players’ scores are computed as the sum of the values of the prize cards they have won.
We remark that due to the tie-breaking rule that we employ, even two-player instances of the game are general-sum. All the Goofspiel instances have {\em limited information}, {\em i.e.}, actions of the other players are observed only at the end of the game. This makes the game strategically more challenging, as players have less information regarding previous opponents' actions.

\paragraph{Leduc}
We use a three-player version of the classical Leduc hold'em poker introduced by~\citet{southey2005bayes}.
In a Leduc game instance with $r$ ranks the deck consists of three suits with $r$ cards each. As the game starts players pay one chip to the pot. There are two betting rounds. In the first one a single private card is dealt to each player while in the second round a single board card is revealed. The maximum number of raise per round is set to two, with raise amounts of 2 and 4 in the first and second round, respectively. 

\paragraph{Battleship}
Battleship is a parametric version of the classic board game, where two competing fleets take turns at shooting at each other. For a detailed explanation of the Battleship game see the work by~\cite{farina2019correlation} that introduced it. Our instance has loss multiplier equal to $2$, and one ship of length $2$ and value $1$ for each player

\subsection{Details about experimental setup}

All experiments were run on a machine with a 16-core 2.80GHz CPU and 32GB of RAM. Fixed points for EFCE dynamics were computed via the Eigen library version 3.3.7~\citep{eigenweb}.

At each time $t$, we let all players in the game pick their \emph{mixed} strategy $\vec{q}^{(i),t}\in\Q$ according to \cref{algo:final}. Each player then observed their own linear utility function.

Note that no randomization is used in the experiments. Indeed, while it would technically be possible to have all players sample and play a deterministic strategy $\vec{\pi}^{(i),t}$ from $\vec{q}^{(i),t}$ and later compute the empirical frequency of the $\vec{\pi}^{(i),t}$, we instead compute the EFCCE gap of the expected empirical frequency directly. This greatly improves the convergence rate in practice.
\fi

\end{document}